\documentclass[11pt,twoside]{article} 
\usepackage{acta-info}
\usepackage{euler}

\newcommand{\vol}{\textbf{4,} 1 (2012) 119--129}   
\newcommand{\amss}{\amssubj{%
Primary: 05C25, secondary: 05C12
}}     
\newcommand{\CRs}{\CRsubj{%
G.2.0
}}   
\newcommand{\keyw}{\keywords{%
Cayley graph, diameter, combinatorial puzzle
}}

\begin{document}

\title{%
Analysis of the picture cube puzzle
}

\maketitle

\oneauthor{%
\href{http://compalg.inf.elte.hu/~bupe}{P\'eter BURCSI} 
}{%
\href{http://compalg.inf.elte.hu}{E\"otv\"os Lor\'and University \\Budapest}
}{%
 \href{mailto:bupe@compalg.inf.elte.hu}{bupe@compalg.inf.elte.hu}
}


\short{%
P. Burcsi
}{%
Analysis of the picture cube puzzle
}

\begin{abstract}
In this paper we give a mathematical model for a game that we call
picture cube puzzle and investigate its properties. The central
question is the number of moves required to solve the puzzle. A
mathematical discussion is followed by the description of computational
results. We also give a generalization of the problem for finite
groups.
\end{abstract}


\section{Introduction}

The picture cube puzzle\footnote{These puzzles are usually sold in
Hungary under the name ``mesekocka'', meaning fairy tale cube.} is a
puzzle that consists of usually 6, 9 or 12 painted wooden cubes that can
be arranged in a rectangular pattern to obtain six different pictures
(``the solutions'') seen on the top of the cubes. Each face of the
cubes contains one piece of one of the six pictures in such a way that
the position of that piece within the large picture is the same for
all six faces. Thus, for example, there is a cube whose faces contain
the upper left corners of the six pictures, another one that contains
the lower right corners, etc.

The cubes are painted so that the solutions can easily be transformed
into each other: for each picture there are cube rotations whose
simultaneous application to all cubes transforms the picture to
another one. For example, imagine that the puzzle is solved and you
can see the picture of the dog on top. Pick up the first row of cubes
holding them together and rotate the whole row around the axis through
the centers of the cubes in that row by 90 degrees. Do this
to all the rows and you obtain the picture of the bear. These row-wise
or column-wise rotations are our allowable moves for the puzzle, see Figure \ref{fig1}.

\begin{figure}
\includegraphics[scale=0.25]{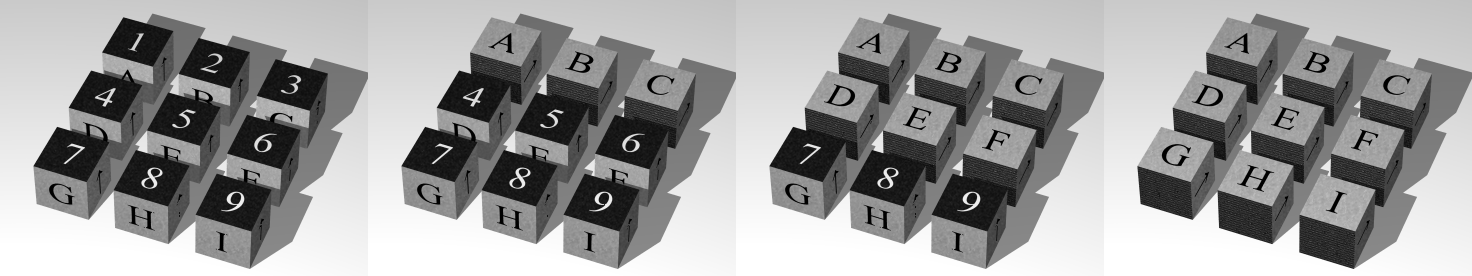}
\caption{A picture cube puzzle. Rotating each row upwards transforms the picture of the numbers  into the picture of the letters. We count this as three moves (there are three rows to rotate).}
\label{fig1}
\end{figure}

The number of allowable moves is twice the sum of the number of rows
and the number of columns in the configuration. No move changes the
position of the cubes within the picture but using the moves we can
change which faces are on top and the orientation of the top
face. A natural question is if we are given an arrangement
of the cubes where the positions are correct but the cubes are
arbitrarily rotated in place, can we solve the puzzle?

For the rest of the paper, we will suppose that one of the solved
pictures is \emph{the} solution configuration, and the cubes are somehow
individually rotated (but kept in place).  There are several
questions that can be raised about the puzzle, we will formalize them
in the following section. What are the configurations that can be reached
from the solution configuration using allowable moves? How many such
configurations exist?  How many moves are needed to reach these configurations
in the worst case/on average?  We will address these questions after
introducing a mathematical model using groups for the puzzle.

These (and other) questions have been raised for more well-known
puzzles, notably Rubik's Cube \cite{dav}. Answering some of them requires
the combination of non-trivial mathematical ideas with large-scale
computer calculations. Diameters for permutation groups have been investigated e.g. in \cite{babai1, babai2}.

The present paper is built up as follows. Section \ref{sec2} gives the mathematical model and a generalization for the puzzle. In Sections \ref{secloc} and \ref{sec4} we discuss solution methods for the cube puzzle. In Section \ref{sec5}, some computational results are presented. We give further research directions and a short summary in Section \ref{sec6}.

We refer to any standard textbook, e.g. \cite{hall}, for basic facts about groups.

\section{Mathematical formulation of the puzzle}\label{sec2}

In order to formalize the puzzle, we need some notation. Let $m$ and
$n$ be two positive integers and imagine $mn$ cubes arranged in an $m$
by $n$ matrix form in front of us on a table. There is a solution
configuration of the cubes when a picture can be seen on the top faces that
we will refer to as "up" face, following standard terminology for
Rubik's Cube. We will refer to the other faces as down, left, right,
front and back faces in the straightforward way. 

There are $2(m+n)$ moves that allow us to modify the configuration of the cubes:
$\textrm{up}_i$ and $\textrm{down}_i$, $(i=1,\ldots, m)$ that rotate the
cubes in row $i$ along the axis through their centers in such a way
that their front faces move to the up or down positions respectively;
and $\textrm{left}_j$ and $\textrm{right}_j$, $(j=1,\ldots, n)$ that
rotate the cubes in column $j$ along the axis through their centers in
such a way that their up faces move to the left or right positions
respectively. We note that two consecutive ``up'' rotations on the same row 
have the same effect as two consecutive ``down'' rotations on that row (similarly
for columns). We count these as two moves---using Rubik's Cube terminology, we
use the quarter-turn metric in the present paper. When these transformations
are counted as one move, we get the half-turn metric.

It is a classical observation in many similar puzzles that sequences
of moves form a group under composition. Allowable moves form a
generator set for that group and we identify the solution of the
puzzle with the identity element. Applying a move to a configuration is
multiplication by a generator from the right. Below we describe the
group and the generators that we will use.

It is a well-known fact of group theory (seen most easily by observing
how the transformations act on the main diagonals) that the rotation
group of the cube is isomorphic to $S_4$, the symmetric group of
degree 4. Thus we can represent any configuration (maybe
 not reachable by legal moves) of the puzzle by a
matrix $T=(g_{ij})\in S_4^{m\times n}$, where $S_4^{m\times n}$ is
itself a group, being the direct product of $mn$ copies of
$S_4$. If we fix a way we represent rotations of a cube by $S_4$ then
there are elements $u, d, l, r\in S_4$ such that turning a single cube
up (resp. down/left/right) corresponds to multiplication (from the
right) by $u$ (resp. $d$/$l$/$r$). Using cycle notation, one
may choose for example $u=(1423)$, $d=(1324)$, $l=(1342)$, $r=(1243)$. Note that all
of them are odd permutations. Then $\textrm{up}_i$ is the
transformation that replaces $g_{ij}$ by $g_{ij} \cdot u$ for
$j=1,\ldots, n$, in other words, $\textrm{up}_i$ is (coordinate-wise)
multiplication in $S_4^{m\times n}$ by an element that has $n$ entries
of $u$ in the $i$th row and the identity element in every other
position, and similarly for the other moves. Thus we may identify
$\textrm{up}_i$ and the other moves by some elements in $S_4^{m \times
n}$. The set of moves in this group will be denoted by $M$.
So $M = \{\textrm{up}_i, \textrm{down}_i \mid 1\leq i\leq m\} \cup \{\textrm{left}_j, \textrm{right}_j\mid 1\leq j\leq n\}$

We may now formulate several questions.  
\begin{itemize} 
\item \textbf{Reachability.} What are the states from which the solution is
reachable using legal moves, or formally: what is the subgroup $H$ of
$S_4^{m\times n}$ generated by $M$? We are also interested in the
order of this subgroup and how membership in this subgroup
(solvability of a configuration) can be decided. Note that reachability between configurations 
 is symmetric since every sequence of moves can be executed "backwards".
\item \textbf{Solution.} Given an element in $H$, what is a sequence of
moves that takes it to the identity (solve the puzzle). We are
interested in human-executable ("easy") and computer-aided ("fast")
techniques as well. What is the shortest sequence of moves, in other
words, what is the shortest product of moves that gives $h\in H$?
\item \textbf{Diameter.} What is the maximum length, taken over elements
$h\in H$, of shortest sequences of moves taking $h$ to the
identity? In other words, how many moves are required in the worst case for solving the puzzle?

\end{itemize}

These questions can be analized by using Cayley graphs that are defined as
follows. 

\begin{definition}
Let $H$ be a group generated by $M$.  The Cayley graph for $H$
and $M$ has $H$ as the set of vertices and there is a directed edge
from $h_1$ to $h_2$ iff for some $m\in M$, $h_1m = h_2$. 
\end{definition}

Solving the
puzzle is the same as finding a path to the identity; and the diameter
of this graph is exactly the length of the longest of all shortest paths from some
$h$ to the identity in the Cayley graph. We return to these questions in the next section.

\subsection{Generalization for arbitrary finite groups}

The mathematical formulation above allows us to generalize the problem
for arbitrary finite groups. Let $G$ be a finite group and let $R, C
\subseteq G$ (row and column moves, respectively). Let $n,m$ be
positive integers and consider the group $G^{m\times n}$ with pointwise multiplication. 
For each $r\in R$ and $1\leq i\leq m$
let $r_i\in G^{m\times n}$ be a matrix which has $n$ entries equal to $r$ in
the $i$th row and the identity element $e\in G$ in other rows. (Here '$r$' stands for row, not to be confused with the $r$ for the original $S_4$ puzzle, where
it is short for right, and is a column move.) For each $c\in
C$ and $1\leq j\leq n$ let $c_j\in G^{m\times n}$ be a matrix which
has $m$ entries $c$ in the $j$th column and the identity element $e\in G$ in
other columns:
\[
r_i = \left( \begin{array}{cccccccc}
e & e & e & \ldots & \ldots  & e & e & e\\
e & e & e & \ldots & \ldots  & e & e &e\\
\vdots & \vdots & \vdots & \vdots &  \ddots & \vdots & \vdots & \vdots\\
e & e & e & \ldots & \ldots  & e &e &e \\
r & r & r & \ldots & \ldots & r &r &r\\
e & e & e & \ldots & \ldots  & e & e &e\\
\vdots & \vdots & \vdots & \vdots & \ddots & \vdots &\vdots & \vdots\\
e & e & e & \ldots & \ldots  & e & e & e\\
\end{array}\right)
\quad
c_j = \left( \begin{array}{cccccccc}
e & e & \ldots & e & c & e & \ldots  & e\\
e & e & \ldots  & e & c & e & \ldots  & e\\
e & e & \ldots  & e & c & e & \ldots  & e\\
\vdots & \vdots & \ddots & \vdots & \vdots & \vdots & \ddots & \vdots \\
\vdots & \vdots & \ddots & \vdots & \vdots & \vdots & \ddots & \vdots \\
e & e & \ldots & e & c & e & \ldots  & e\\
e & e & \ldots  & e & c & e & \ldots  & e\\
e & e & \ldots  & e & c & e & \ldots  & e\\
\end{array}\right)
\]

\begin{definition}
Given $G$, $R$, $C$, $m$ and $n$ as above, the set of row moves is $\{r_i\mid r\in R, 1\leq i\leq m\}$,
the set of column moves is $\{c_j\mid c\in C, 1\leq j\leq n\}$. The set of moves is their union
$M=M(G,R,C,m,n) = \{r_i\mid r\in R, 1\leq i\leq m\} \cup \{c_j\mid
c\in C, 1\leq j\leq n\}$.
\end{definition}

\begin{definition}
Given $G$, $R$, $C$, $m$ and $n$ as above, the set of reachable configurations
is $H = H(G, R, C, m, n) = \langle M \rangle \leq G^{m \times n}$. 
\end{definition}

We can ask what the subgroup $H$ is, how we can give a product of moves
for an $h\in H$ and what the diameter of the Cayley graph
for $H$ and $M$ is. These questions seem hard to answer in general,
so we only focus on a few special cases that will be useful in the $S_4$ case which 
is a special case with $G = S_4$, $R=\{\textrm{up}, \textrm{down}\}$, 
$C=\{\textrm{left}, \textrm{right}\}$.

\subsection{Abelian groups}

If $G$ is Abelian, then so is $G^{m\times n}$, meaning that the order
in which we perform the moves has no effect on their product. Therefore,
instead of sequences of moves, we can speak of sets of moves. Denote
by $a_i$ (resp. $b_j$) the product of moves performed on the $i$th row
(resp. $j$th column). Note that $a_i$ is not necessarily an element in
$R$. Then the product of all the moves performed has the matrix form
$T = (g_{ij})$ where $g_{ij} = a_i b_j$. We claim that the first row
and the first column determine all other elements in $T$. To see this,
let $i, j > 1$, then $g_{ij} = a_i b_j = (a_1 b_j) (a_i b_1) (a_1
b_1)^{-1} = g_{1j}g_{i1}g_{11}^{-1}$.

Consider the special case when $R=C=G$.
\begin{theorem}
If $G$ is Abelian, $R=C=G$ and $n,m\in\mathbb{N}^{+}$, then the reachable configurations are
$H=\{ (a_ib_j)_{i=1,\ldots, m, j=1,\ldots, n} \mid a_i\in G, b_j\in G \}$, 
and we have $|H| = |G|^{m+n-1}$.
\end{theorem}

\begin{proof}
For any configuration in $G^{m \times n}$, we can solve the first row by using column operations, then solve the remaining elements of the first column by row operations. If the rest is not solved at this point,
then the configuration is not in $H$, because the remaining elements of the configuration are uniquely determined by the first row and the first column for elements of $H$. 
Thus $H$ has configurations of the form $(a_ib_j)_{i=1,\ldots, m, j=1,\ldots, n}$ if we let e.g. $a_i=g_{i1}$ and $b_j = g_{1j}/g_{11}$. It is also clear that $H$ contains every element of this form.
The number of required moves to solve the puzzle is $m+n-1$ in the worst case. The
order of $H$ is $|G|^{m+n-1}$, because that is the number of ways we can adjust the elements in the first row and the first column.
\end{proof}

If we weaken the conditions but require that both $R$ and $C$ generate
$G$, then the solution method can be the same, but for solving the
first row and column, a sequence of moves is required for each
element. 

We note that at least one special case of this Abelian version is part of mathematical folklore: when the group is the two-element group and the only moves are multiplication by the generator. This is often told with coins (and sometimes with lamps) arranged in a matrix form and allowed moves being the simultaneous turning over of entire rows or columns. The question is how we can tell if an initial configuration can be transformed into the "all heads" configuration.

The discussion gets more complicated when at least one of $R$ and $C$
generates only a nontrivial subgroup of $G$, but since we do not need this
for the cube puzzle case, we omit the analysis of this case for brevity.

\subsection{Simple groups}

Another special case is when G is a non-abelian simple group.  We
investigate this case because the method used for solving it -- the
method of commutators -- is also useful for the cube puzzle. 

\begin{definition}
Let $G$ be a group, $g_1, g_2\in G$. The commutator of $g_1$ and $g_2$ is
$[g_1, g_2] = g_1g_2g_1^{-1}g_2^{-1}$. This is the identity element if and only if $g_1$ and $g_2$ commute. Note that in some sources in the literature, the inverses come first in the definition.
\end{definition}

\begin{lemma}
Let $G$ be a group, $R, C\subseteq G$ and $m, n$ as above, $r\in R, c\in C$, $1\leq i\leq m$, $1\leq j\leq n$. Then multiplying a configuration by $[r_i, c_j]$ only affects the entry at position $i,j$.
\end{lemma}

\begin{proof}
Denote the entries of the configuration by $g_{i'j'}$.
The transformation $t = [r_i, c_j] = r_i c_j r_i^{-1} c_j^{-1}$ can only
alter elements in the $i$th row and the $j$th column. But if $j'\neq
j$, then $c_j$ does not modify the elements in column $j'$, hence the
effect of $t$ on $g_{ij'}$ is multiplication by $r_ir_i^{-1} = e$, the
identity. Similarly, $g_{i'j}$ is left intact by $t$ if $i'\neq
i$. Thus, the only element that can be affected is $g_{ij}$. The
effect on this element is multiplication by $[r,c]=rcr^{-1}c^{-1}$. See Figure \ref{fig2} for an illustration.
\end{proof}

\begin{figure}
\includegraphics[scale=0.2]{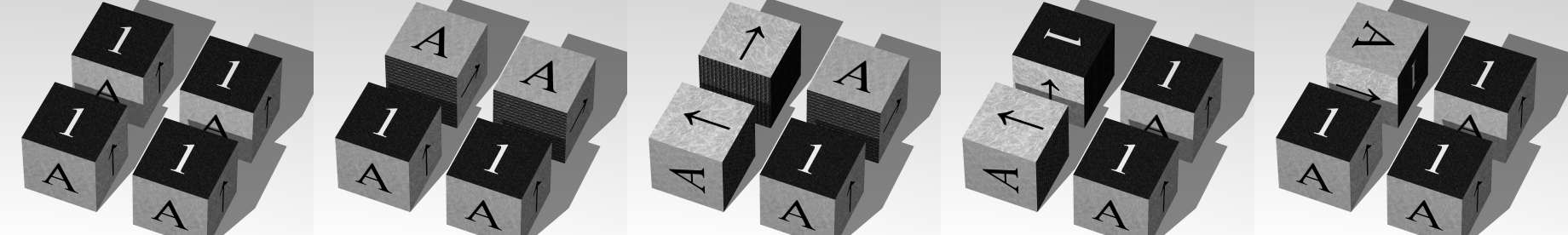}
\caption{The figure illustrates how a commutator of a row move and a column move modifies only one entry. The sequence of moves up$_1$, left$_1$, down$_1$, right$_1$ cancels for all entries except at position $1, 1$.}
\label{fig2}
\end{figure}

Using this method we can individually modify the elements of
$(g_{ij})$ by commutators of the form $[r,c]$ or $[c,r]$. More
generally, using sequences of moves we can individually modify
$g_{ij}$ by $[r,c]$ or $[c,r]$ for $r\in \langle R\rangle$, $c\in
\langle C \rangle$. Hence we have the following result.

\begin{theorem}
Let $G$ be a finite non-abelian simple group, $R=C=G$, $m,n\in \mathbb{N}^{+}$. Then 
$H = H(G, R, C, m, n) = G^{m\times n}$, in other words, every configuration is solvable. 
\end{theorem}

\begin{proof}
The commutator subgroup of $G$, denoted by $G'$ is the subgroup
generated by all commutators $[g_1, g_2]$ for $g_1,g_2 \in G$. 
By the above lemma, we can alter any individual entry of a configuration
by any element in $G'$. But $G'$ is
always normal in $G$, hence if $G$ is a finite simple non-abelian
group, then $G'=G$. Therefore any configuration can be solved, by 
individually solving all entries.
\end{proof}

Note that the theorem also holds if we only assume $\langle R\rangle = \langle C
\rangle = G$ -- note that for $r\in R$, $r^{-1} \in \langle R \rangle$
by the finiteness of $G$. Also note that for Abelian groups, the commutator method is useless,
since we have $G'=\{e\}$. 

If $G$ is a finite simple non-abelian group and $\langle R\rangle =
\langle C \rangle = G$, then the diameter of the Cayley graph can be
bounded both from below and above for general $m$ and $n$ using
constants that depend only on $G, R, C$. The number of vertices of the graph
is $|G|^{mn}$, the graph is $(m|R|+n|C|)$-regular. From this we get
the lower bound $\log (|G|^{mn}) / \log(m|R|+n|C|)-1 = \textrm{const}_1
mn/\log(m|R|+n|C|)$. For the upper bound, note that we can solve each
configuration in $\textrm{const}_2$ moves. Thus the number of moves
required is at most $\textrm{const}_2 \cdot mn$: 
\begin{equation}
\frac{\textrm{const}_1 \cdot mn}{\log(m|R|+n|C|)} \leq \textrm{diam}(\textrm{Cayley}(H, M)) \leq \textrm{const}_2\cdot mn.
\end{equation}
The gap between the lower
and upper bounds is logarithmic in $m$ and $n$. It would be
interesting to reduce this gap. We conjecture that the diameter is
closer to the upper bound.

\section{Local solution method for the cube puzzle}
\label{secloc}

We return to the analysis of the original cube puzzle, where $G=S_4$ and
the row moves are turning up and down ($R = \{u, d\}$),
the column moves are turning left and right ($C = \{l, r\}$). The subgroup in $S_4^{m\times n}$ generated by all moves is again denoted by $H$.
In this section we provide a decision method for solvability and present a method for solving using commutators. The method is simple and 
easy to perform for humans. We suppose that some configuration $s\in S_4^{m \times n}$ is given and 
we want to decide if it is in $H$, and if it is, we want to solve it.

We break the solution of the puzzle into two parts. First, we look for a sequence of moves $t$ that transforms $s$ to $st \in A_4^{m \times n}$, that is, we try to make every entry in $s$ into an even permutation. If this is impossible, then $s\not\in H$. If it is possible, we will solve the puzzle by using commutators. 

For the first part note that since an up, down, left or right move toggles the parity of the affected elements (seen as permutations in $S_4$), the task of making every entry an even permutation can be reduced to solving a puzzle for the two-element multiplicative group $\{-1, 1\}$ where the legal row and column operations are multiplication of the row or column by $-1$.  The Abelian analysis in the previous section shows that this can be examined by solving the first row and the first column and seeing if the rest is solved. If the rest is not solved, there is no solution for the puzzle. If there is a solution, it is found in at most $m+n-1$ steps.
To perform the first part, the player has to remember which rotations correspond to even and odd permutations, but this is not too hard.

In the second part, we individually rotate all $mn$ cubes to their solved positions, using the method of commutators. One readily verifies that commutators of the form $[u_i, l_j]$, $[u_i, r_j]$ etc. (altogether 8 types for a fixed pair $i, j$) are sufficient to transform any element of $A_4$ to the identity in at most 8 steps (two commutators suffice). This part then requires a worst case $8mn$ moves. The overall number of steps needed to solve any solvable configuration using this local method is thus at most $m+n-1+8mn$.

We summarize the above discussion in the following theorem.

\begin{theorem}
Let $m,n\in\mathbb{N}^{+}$. Then $|H| = |H(S_4, R, C, m, n)| = 2^{m+n-1} \cdot 12^{mn}$. The diameter of the Cayley graph of $H$ and $M$ is at most $8mn+m+n-1$.
\end{theorem}

\section{Solution method using subgroups}\label{sec4}
Another method borrowing ideas from the Rubik's Cube literature \cite{koc, th}
 is the method of subgroups generated by restricted moves. The method of the previous section can be considered as a special case of this method, but in general, the method is not intended for human execution. The method is most useful when the search space for finding the overall shortest path in the Cayley graph is too large, but a computer-aided search using subgroups finds shorter paths than the commutator-based method.

Take any nontrivial subgroup $H_1$ of $H$, in practice a subgroup with comparable index and order is preferable. Let the input for the method be an $h\in H$, and we proceed in two parts. First find the shortest sequence of moves that takes $h$ into $H_1$, then solve the problem for $H_1$. One can also use a chain of subgroups, but we only consider the case of one subgroup. The subgroups of interest here can admit the following form. Let $K$ be a subset of $\{1,\ldots, m\}$. We allow transformations generated by all possible moves with the restriction that for rows with index $i\in K$, only double moves $u_i^2 = d_i^2$ are allowed (in the half-turn metric this is extremely useful, since these moves count as one move). This method can be used to find a shorter move sequence than the naive method that is not necessarily optimal. The optimal choice of $K$ is a nontrivial issue, in practice, computer experiments can be used to choose a good size for $K$.

\section{Computational results}\label{sec5}
For small values of $m$ and $n$ we can use a computer to determine the diameter or other properties of the Cayley graph in our problem. The number of possible configurations grows exponentially in $nm$, so even for moderate values, we quickly run into memory storage problems. We represented configurations in a compact form using base 24 integer numbers. For the determination of the diameter and average distances, a breadth-search is performed, where the neighbors of a configuration are calculated using pre-stored 24-element rotation arrays -- we listed in advance how the moves transform cube positions.

The tables summarize the results. Table \ref{tab1} lists the number of reachable states, the average and median distance and the maximal distance for various puzzle sizes. Table \ref{tab2} has details about the $3\times 2$ case, listing how many configurations can be found on individual levels of the tree.

\begin{table}
\begin{center}
\begin{tabular}{|c|c|c|c|c|c|}
\hline
Size & Diameter & Average & Median & Local est. & Nr. of conf. \\
\hline
$1\times 1$ & $4$ & $2.17$ & $2$ & $9$ & $24$\\
\hline
$2\times 1$ & $7$ & $4.44$ & $5$ & $18$ & $576$ \\
\hline
$3\times 1$ & $10$ & $6.59$ & $7$ & $27$ & $13824$ \\
\hline
$4\times 1$ & $12$ & $8.59$ & $9$ & $36$ & $331776$ \\
\hline
$2\times 2$ & $12$ & $7.82$ & $8$ & $35$ & $16588$\\
\hline
$5\times 1$ & $15$ & $10.69$ & $11$ & $45$ & $7962624$\\
\hline
$6\times 1$ & $17$ & $12.65$ & $13$ & $54$ &  $191102976$\\
\hline
$3\times 2$ & $14$ & $10.54$ & $11$ & $52$ &  $47775744$\\
\hline
\end{tabular}
\end{center}
\caption{This table lists the maximal/average/median number of optimal moves needed to solve a puzzle of the given size. It also lists the number of moves that the local method from Section \ref{secloc} takes and the number of reachable configurations.}
\label{tab1}
\end{table}

\begin{table}
\begin{center}
\begin{tabular}{|c|c|c|c|c|c|}
\hline
Distance & 0 & 1 & 2 & 3 & 4\\
\hline
Number of points   & 1 & 10 & 69 & 456 & 2846 \\
\hline
\hline
Distance & 5 & 6 & 7 & 8 & 9\\
\hline
Number of points   & 16208 & 84428 & 395566 & 1622641& 5536264  \\
\hline
\hline
Distance & 10 & 11 & 12 & 13 & 14  \\
\hline
Number of points &  13587945 & 17558644 & 8100138 & 843444 & 27084 \\
\hline
\end{tabular}
\end{center}
\caption{The number of reachable configurations that are $0, 1,\ldots, 14$ moves away in the $3\times 2$ puzzle.}
\label{tab2}
\end{table}

\section{Summary and further directions} \label{sec6}
 We gave a model for the
picture cube puzzle which allowed us to answer some naturally arizing
combinatorial questions. Mathematically, the study  of
general $G, R, C$ and the asymptotic analysis of the diameter is the
planned continuation of the present work. From a computational point
of view, the further investigation of the exact values of the diameter
for small $m,n$ in the cube puzzle is our future plan.

Replacing the 2-dimensional configuration by higher dimensional ones
is also a possible extension.

\section*{Acknowledgements}
The Project is supported by the   European Union and co-financed by the European Social Fund (grant agreement no. T\'AMOP 4.2.1/B-09/1/KMR-2010-0003).

\bigskip
\rightline{\emph{Received: March 15, 2012 {\tiny \raisebox{2pt}{$\bullet$\!}} Revised: April 23, 2012}} 
\end{document}